\documentclass[11pt,reqno]{amsart}

\usepackage[T1]{fontenc}
\usepackage{lmodern}
\usepackage[expansion=false]{microtype}
\usepackage[a4paper,textwidth=150mm,textheight=242mm,centering,
  headheight=14pt,headsep=7mm,footskip=11mm]{geometry}
\usepackage{amsmath,amssymb,amsthm,mathtools}
\usepackage{hyperref}
\hypersetup{colorlinks=true,linkcolor=blue,citecolor=blue,urlcolor=blue,
  pdftitle={Wehrl-type entropy problem for compact connected semisimple Lie groups},
  pdfauthor={},
  pdfsubject={Haar measure formulation, real Husimi moments, and direct entropy variation}}
\usepackage{bookmark}
\usepackage{braket}
\allowdisplaybreaks[2]
\numberwithin{equation}{section}

\theoremstyle{plain}
\newtheorem{theorem}{Theorem}[section]
\newtheorem{proposition}[theorem]{Proposition}
\newtheorem{corollary}[theorem]{Corollary}
\newtheorem{lemma}[theorem]{Lemma}
\theoremstyle{remark}
\newtheorem{remark}[theorem]{Remark}

\newcommand{\C}{\mathbb C}

\newcommand{\R}{\mathbb R}

\newcommand{\gfrak}{\mathfrak g}
\newcommand{\ii}{\mathrm i}
\newcommand{\dd}{\,\mathrm d}
\newcommand{\Id}{\mathrm I}
\newcommand{\tr}{\operatorname{Tr}}
\newcommand{\Ad}{\operatorname{Ad}}

\newcommand{\norm}[1]{\lVert#1\rVert}
\newcommand{\ip}[2]{\langle#1,#2\rangle}
\newcommand{\proj}[1]{\lvert#1\rangle\langle#1\rvert}
\newcommand{\SW}{S_{\mathrm W}}
\newcommand{\wvec}{w}
\newcommand{\wopt}{\wvec_*}

\title[Wehrl-type entropy for compact semisimple groups]{Wehrl-type entropy problem for compact connected semisimple Lie groups}
\date{September 10, 2026}

\author{Haonan Zhang}
\address{Department of Mathematics, University of South Carolina, Columbia, SC 29208, USA}
\email{haonanzhangmath@gmail.com}
\date{\today}

\begin{document}
\raggedbottom
\begin{abstract}
Let $G$ be a compact connected semisimple Lie group, and let $\pi:G\to U(V_\lambda)$ be a
finite-dimensional irreducible unitary representation associated with the highest weight $\lambda$.
We prove that coherent projectors are the unique minimizers of the Wehrl entropy
over all density matrices on $V_\lambda$. Similar extremal properties hold for power functions 
in place of the entropy function. The proof uses a second variation at the extremizer by perturbation
in directions associated with Killing fields, similar to the strategy of Frank and Lieb used in \cite{FrankLieb}. Then the problem reduces to the uncertainty principle characterizing coherent states.
This solves the Wehrl-type entropy problem for arbitrary compact connected semisimple Lie groups.
\end{abstract}
\maketitle

\section{Introduction}\label{sec:intro}

Wehrl \cite{Wehrl,Wehrl79} defined a classical analog of the von Neumann entropy
$S(\sigma)=-\tr(\sigma\log\sigma)$ of density operators using coherent states. For a density operator $\sigma$ on $L^2(\R^n,\dd x)$, consider the Husimi function
\begin{equation*}
 Q_\sigma(q,p)=\langle q,p|\sigma|q,p\rangle,
\end{equation*}
where $\ket{q,p}\in L^2(\R^n,\dd x), q,p\in\R^n$ are the Glauber coherent states \cite{Glauber}
\[
 \ket{q,p}(x)=(\pi\hbar)^{-n/4}
 \exp\left(-\frac{|x-q|^2}{2\hbar}
       +\frac{\ii}{\hbar}p\cdot\left(x-\frac q2\right)\right),\qquad x\in \R^n.
       \]
Here $\hbar>0$ is the reduced Planck constant, and we may choose $\hbar=1$. The Husimi function is a probability density on phase space $(\R^{2n},\mathrm d q\,\mathrm d p/(2\pi\hbar)^n)$, and 
the Wehrl entropy is defined as the classical differential entropy of its Husimi function:
\[
 \SW(\sigma)=-\int_{\R^{2n}}Q_\sigma\log Q_\sigma\,
                  \frac{\mathrm d q\,\mathrm d p}{(2\pi\hbar)^n}.
\]
Throughout, we use $0\log0=0$ and take inner products to be conjugate-linear
in the first argument and linear in the second.

In general, the von Neumann entropy $S(\sigma)$ is nonnegative and vanishes when $\sigma$ is a pure state. The classical differential entropy of an arbitrary probability 
density can be negative and unbounded from below. However, the Wehrl entropy establishes a different feature. 
Since $0\le Q_\sigma\le 1$, $\SW(\sigma)$ is nonnegative.
Wehrl conjectured \cite{Wehrl79} that coherent
projectors $\ket{q,p}\bra{q,p}$ minimize this entropy $\SW$, and this was solved by Lieb \cite{Lieb}
using sharp Hausdorff--Young and Young convolution inequalities. In fact, he proved the analogous results
 for power functions $x^p, p>1$, and then derived the entropy result by differentiating at $p=1$.
Carlen \cite{Carlen} gave another proof using logarithmic Sobolev inequalities, and his proof works for power functions $x^p, 0<p<1$ as well.
He also obtained the equality condition for the first time.
Luo \cite{Luo} gave an alternative proof of the original Wehrl entropy inequality based on hypercontractivity.

The above problem has a group-theoretic formulation, and coherent states can be defined
for more general Lie groups \cite{Gilmore1972,Perelomov1972,Perelomov}. Wehrl's original problem concerns the Heisenberg group, and it is natural to ask if the same extremal phenomenon occurs for more Lie groups. In particular, Lieb \cite{Lieb} conjectured the same result for the Bloch
coherent states associated with $SU(2)$. After a number of partial results \cite{Bodmann,Gnutzmann,Schupp}, this was solved by Lieb and Solovej
\cite{LS2} by determining the minimum output entropy of a class of quantum channels and a semiclassical limit argument.
Their proof applies to general convex functions, including the power functions and $x\log x$.
Later, the same authors \cite{LS} extended the results to symmetric representations of $SU(N)$ for all $N$.
We refer to \cite{SchuppSurvey} for more discussions. 

Lieb and Solovej \cite{LS11} also studied the Wehrl-type entropy problem
for the holomorphic discrete series of $SU(1,1)$ and proved the case of integer power functions; see also \cite{Bandyopadhyay,Bayart,BGS}. The corresponding
Wehrl-type problem was finally solved by Kulikov
\cite{Kulikov} using the hyperbolic isoperimetric inequality for level
sets of weighted holomorphic functions.
Kulikov also characterized equality for strictly convex functions.
For related inequalities, equality cases, and
stability results for Heisenberg, $SU(2)$, $SU(1,1)$, and symmetric $SU(N)$
coherent states, see \cite{NT,Frank,KNOT,FNT,NRT,NRTsecond}.

For arbitrary compact connected semisimple groups, Sugita \cite{Sugita} proved the
sharp inequalities for integer powers $p\ge2$ by tensor-product methods.
See also \cite{HZ} for more extensions.

\medskip

In this paper, we provide a solution to the Wehrl-type entropy problem for
arbitrary compact connected semisimple Lie groups. To state the results, let us recall the formulation of the problem first.

Let $G$ be a compact connected semisimple Lie group, and let
$\pi:G\to U(V_\lambda)$ be a finite-dimensional irreducible unitary
representation corresponding to the highest weight $\lambda$.  Write $d=\dim V_\lambda$,
fix a unit highest-weight vector $v_\lambda\in V_\lambda$, and let $\mathrm d g$ be
normalized Haar measure on $G$ so that $\int_G\mathrm d g=1$.
For $g\in G$, set $ v_g=\pi(g)v_\lambda,$ and consider the coherent projectors
\[
 P_g=\proj{v_g}=\pi(g)\proj{v_\lambda}\pi(g)^*.
\]
They satisfy the resolution of identity
\[
 \int_G P_g\dd g=\frac1d\Id.
 \]
 Then for any density matrix $\sigma$ on $V_\lambda$, the Husimi function
 \begin{equation}\label{eq:husimi}
 Q_\sigma(g)=\tr(\sigma P_g)=\ip{v_g}{\sigma v_g}
\end{equation}
is a probability density function on $(G, d\cdot \dd g)$. We define the Wehrl-type entropy by
\begin{equation}
 \SW(\sigma):=-d\int_G Q_\sigma(g)\log Q_\sigma(g)\dd g.
\end{equation}
When $\sigma$ is a pure state $\proj{\wvec}$,
we write $Q_{\wvec}=Q_{\proj{\wvec}}$, so that $Q_{\wvec}(g)=|\ip{v_g}{\wvec}|^2$. The Wehrl-type entropy conjecture asserts that
the minimum of $\SW(\sigma)$ over all density matrices $\sigma$ on $V_\lambda$ is attained at a coherent state $\proj{v_g}$ for some (and for all) $g\in G$.
The generalized problem concerns functions beyond $-x\log x$, as well as the equality conditions.

The following lemma states that coherent states have minimal uncertainty; for a proof see \cite{DelbourgoFox} and
\cite[Appendix~A, proof of Proposition~A.1]{HZ}.

\begin{lemma}\label{lem:uncertainty}
Let $(X_j)_{j=1}^s$ be an orthonormal basis for the negative Killing form
on the Lie algebra $\gfrak$ of $G$, and set $T_j=-\ii\,\mathrm d\pi(X_j)$.
For every unit vector $\wvec\in V_\lambda$, the uncertainty
\begin{equation}\label{eq:uncertainty}
 \mathcal U(\wvec):=\sum_j\left(\ip{\wvec}{T_j^2\wvec}-\ip{\wvec}{T_j\wvec}^{\,2}\right)
\end{equation}
satisfies
\begin{equation}\label{eq:uncertaintybound}
 \mathcal U(\wvec)\ge\mathcal U_{\min}:=\mathcal U(v_\lambda).
\end{equation}
Equality holds if and only if $\wvec=e^{\ii\theta}v_h$ for some
$\theta\in\R$ and $h\in G$.
\end{lemma}

Our main result is as follows.

\begin{theorem}\label{thm:main}
Let $\Phi:[0,1]\to\R$ be either $\Phi(x)=x^p$ for a fixed real $p>0$, $p\ne1$,
or $\Phi(x)=x\log x$. For every unit vector $\wvec\in V_\lambda$,
\begin{equation}\label{eq:functional}
 \int_G\Phi(|\langle v_g,\wvec\rangle|^2)\dd g
 \le \int_G\Phi(|\langle v_g,v_\lambda\rangle|^2)\dd g.
\end{equation}
For $\Phi(x)=x^p$ with $0<p<1$, the inequality is reversed.
In all cases, equality holds if and only if
$\wvec=e^{\ii\theta}v_h$ for some $\theta\in\R$ and $h\in G$.
\end{theorem}

For convex $\Phi$, the inequality for general density matrices reduces to
pure states; the same holds for concave $\Phi$ with the inequality reversed.
Strict convexity or strict concavity also gives the corresponding equality condition.
See for example \cite[Proof of Corollary~14]{Frank}.

\begin{corollary}\label{cor:density}
Let $\Phi(x)$ be $x\log x$, $x^p$ with $p>1$, or $-x^q$ with
$0<q<1$. Then, for every density matrix $\sigma$ on $V_\lambda$,
\begin{equation}\label{eq:densityfunctional}
 \int_G\Phi(Q_\sigma(g))\dd g
 \le \int_G\Phi(Q_{v_\lambda}(g))\dd g.
\end{equation}
Equality holds if and only if $\sigma=P_h$ for some $h\in G$.
In particular,
\begin{equation}\label{eq:entropy}
 \SW(\sigma)\ge\SW(\proj{v_\lambda}).
\end{equation}
\end{corollary}

The idea of the proof is to analyze an optimizer under perturbations in
directions associated with Killing fields.
Suppose that $\wopt$ is a unit vector maximizing
$\int_G\Phi(Q_{\wvec}(g))\dd g$, with value $M$, for $\Phi(x)=x^p$, $p>1$.
Let $T_j=-\ii\,\mathrm d\pi(X_j)$, where $(X_j)_j$ is an orthonormal
basis for the negative Killing form on $\gfrak$. We consider the perturbation
\[
 \gamma_j(t)=\frac{\wopt+tT_j\wopt}{\norm{\wopt+tT_j\wopt}},\qquad |t|\text{ small}.
\]
Homogeneity reduces the problem to the scalar functions
\[
 \phi_j(t)=\int_G|\langle v_g,\wopt+tT_j\wopt\rangle|^{2p}\dd g
              -M\norm{\wopt+tT_j\wopt}^{2p},
\]
that are non-positive by definition of $\wopt$ and $M$. The choice of $\wopt$ also implies $\phi_j''(0)\le0$ for each $j$.
On the other hand, some identities for Casimir operators yield
\[
 \sum_j\phi_j''(0)=4p(p-1)M
 \bigl(\mathcal U(\wopt)-\mathcal U_{\min}\bigr)\ge 0.
\]
The final inequality follows from the uncertainty principle in Lemma~\ref{lem:uncertainty}. Therefore, we must have
$\mathcal U(\wopt)=\mathcal U_{\min}$ and thus $\wopt$ is coherent.

\medskip

The proof strategy is similar to the one used by Frank and Lieb in
\cite{FrankLieb} to prove sharp Hardy--Littlewood--Sobolev inequalities
on the Heisenberg group: starting from an optimizer, they examine second
variations in a similar way, then sum the resulting stability inequalities
and combine them with a spectral comparison to force the optimizer to
have the desired form. The author thanks Rupert Frank for pointing this out.

\medskip

Section~\ref{sec:prelim} recalls the Lie-theoretic preliminaries and some Casimir identities.
Section~\ref{sec:variations} collects some technical computations in the second variation argument.
In Section~\ref{sec:proof} we prove the main results.
Some technical details are deferred to Appendices~\ref{app:smallpowers}
and~\ref{app:regularity}.

\subsection*{Acknowledgments} The author is partially supported by NSF grant DMS-2453408 and acknowledges the use of OpenAI’s GPT-5.6 and GPT-6. 
He is grateful to Rupert Frank for valuable comments on an earlier version of the paper, and to Paata Ivanisvili for helpful discussion. 
The author thanks Jan Philip Solovej for hosting him during a visit to Copenhagen in 2022 and for helpful discussion. He thanks Elliott Lieb for encouragement. 
Part of the work was finished during a visit to Texas, and he is grateful to Tao Mei and Ping Zhong for kind hospitality.

\section{Preliminaries on Lie groups and Lie algebras}\label{sec:prelim}

For background on highest-weight representations, see \cite{Humphreys,Knapp}.
For the trivial representation, $d=1$ and $Q_\sigma(g)=1$ for every
$g\in G$, so every moment is one and the entropy is zero.
Theorem~\ref{thm:main} is immediate in this case.  We assume $\lambda\ne0$ below.

\subsection{Basic notations}
Let $G$ be a compact semisimple connected Lie group and $\gfrak$ its Lie algebra.
Consider the inner product defined by the negative Killing form
\begin{equation}\label{eq:killingB}
       B(X,Y)=-\operatorname{Kill}(X,Y)
 =-\operatorname{Tr}
       (\operatorname{ad}X\operatorname{ad}Y).
\end{equation}
Since $\gfrak$ is compact and semisimple, $B$ is positive definite and
$\Ad(G)$-invariant.   For $X\in\gfrak$, define the self-adjoint operator
\[
 T(X)=-\ii\,\mathrm d\pi(X).
\]
Choose a $B$-orthonormal basis $X_1,\ldots,X_s$ of $\gfrak$ and consider the self-adjoint operators
\[
 T_j=T(X_j),\qquad 1\le j\le s.
\]
Write $h_j(t)=\exp_G(-tX_j)$, $t\in\R$, for the one-parameter subgroup
of $G$ generated by $-X_j$; it acts on the representation space $V_\lambda$ as
$\pi(h_j(t))=e^{-\ii tT_j}$.

Set
\[
 c_\lambda:=\sum_j\ip{v_\lambda}{T_j^2v_\lambda},\qquad
 a_\lambda:=\sum_j\ip{v_\lambda}{T_jv_\lambda}^{\,2}.
\]
Then $\mathcal U_{\min}=c_\lambda-a_\lambda$. The explicit values are
$c_\lambda=(\lambda,\lambda+2\varrho)$ and $a_\lambda=(\lambda,\lambda)$, where
$\varrho$ is half the sum of the positive roots and parentheses denote
the inner product on weights induced by $B$ 
\cite[\S11.3, p.~131, equation~(11.3.5)]{Perelomov}.

According to Schur's lemma, one has
\begin{equation}\label{eq:normalization}
 \int_G P_g\dd g=\frac1d\Id,\qquad
 \int_G Q_\sigma(g)\dd g=\frac1d.
\end{equation}

\subsection{Casimir identities}

We collect some identities of Casimir operators that will be used later.

\begin{lemma}\label{lem:casimir}
For every $g\in G$ and every $w\in V_\lambda$, we have
\begin{equation}\label{eq:casimir}
 \sum_j T_j^2=c_\lambda\Id,
\end{equation}
and
\begin{equation}\label{eq:square}
 \sum_j\ip{v_g}{T_jw}^{\,2}=a_\lambda\,\ip{v_g}{w}^{\,2}.
\end{equation}
\end{lemma}
Equation \eqref{eq:casimir} is Schur's lemma for the Casimir operator.
For \eqref{eq:square}, apply Schur's lemma to the $G$-invariant operator
$\sum_jT_j\otimes T_j$ on the irreducible Cartan component
$V_{2\lambda}\subset V_\lambda^{\otimes2}$ containing $v_g^{\otimes2}$.
Indeed,
\[
 \sum_jT_j\otimes T_j
 =\frac12\sum_j\left[(T_j\otimes\Id+\Id\otimes T_j)^2
       -T_j^2\otimes\Id-\Id\otimes T_j^2\right].
\]
The scalar is therefore
$\tfrac12(c_{2\lambda}-2c_\lambda)=(\lambda,\lambda)=a_\lambda$. For details, see for example \cite[Section~III]{Sugita}
 and \cite[Appendix A]{HZ}.

\section{Second variations under perturbations}\label{sec:variations}

Fix $w\in V_\lambda$ and we shall use the perturbation $w+tT_jw$. For this, we write
\[
 A=\ip{v_g}{w},\qquad B_j=\ip{v_g}{T_jw},\qquad C_j=\ip{v_g}{T_j^2w}.
\]
Note that
 $|A+tB_j|^2=|A|^2+2t\Re(\overline A B_j)+t^2|B_j|^2$.

\begin{lemma}\label{lem:scalarvariation}
Suppose $\Phi:[0,\infty)\to\R$ is such that
$z\mapsto\Phi(|z|^2)$ is $C^2$ on $\C$. For $A\ne0$,
\begin{align}\label{eq:affineindividual}
 \left.\frac{\mathrm d^2}{\mathrm dt^2}\right|_0\Phi(|A+tB_j|^2)
 =2\Phi'(|A|^2)|B_j|^2
   +4\Phi''(|A|^2)\bigl(\Re(\overline A B_j)\bigr)^2.
\end{align}
The right-hand side extends continuously across $A=0$. Consequently,
\begin{align}\label{eq:generalvariation}
 \sum_j\left.\frac{\mathrm d^2}{\mathrm dt^2}\right|_0
 \int_G\Phi(|\ip{v_g}{w+tT_jw}|^2)\dd g
=\int_G\left[4a_\lambda|A|^4\Phi''(|A|^2)
 +2c_\lambda|A|^2\Phi'(|A|^2)\right]\dd g.
\end{align}
\end{lemma}
\begin{proof}
The first formula \eqref{eq:affineindividual} is the ordinary chain rule.
The $C^2$ hypothesis on the radial function implies that
$\Phi'(x)$ has a finite limit and $x\Phi''(x)\to0$ as $x\downarrow0$,
by comparing its radial and tangential second derivatives.
This justifies the asserted extension.

For the integrated variation \eqref{eq:generalvariation}, we may differentiate under the integral 
thanks to compactness of $G$ and the $C^2$ assumption.  So the left-hand side becomes
\begin{equation}\label{eq:integral}
\sum_j\int_G \left[ 2\Phi'(|A|^2) |B_j|^2
   +4\Phi''(|A|^2)\bigl(\Re(\overline A B_j)\bigr)^2\right]\dd g.
\end{equation} 
To compute this, we shall use Lemma \ref{lem:casimir} together with the following observation. 

Recall that $e^{-\ii tT_j}$ acts as $h_j(t)=\exp_G(-tX_j)\in G$. So the invariance of Haar measure gives 
\[
 U_j(t):=\int_G\Phi(|\ip{v_g}{e^{\ii tT_j}w}|^2)\dd g
        =\int_G\Phi(|\ip{v_{h_j(t)g}}w|^2)\dd g
        =\int_G\Phi(|A|^2)\dd g.
\]
Thus $U_j$ is constant and $U_j''(0)=0$.
Since
\[
 \ip{v_g}{e^{\ii tT_j}w}
 =A+\ii tB_j-\frac{t^2}{2}C_j+o(t^2),
\]
\[
 |\ip{v_g}{e^{\ii tT_j}w}|^2
 =|A|^2-2t\Im(\overline A B_j)
 +t^2\bigl(|B_j|^2-\Re(\overline A C_j)\bigr)+o(t^2),
\]
the chain rule gives
\[
 0=U_j''(0)=\int_G\left[
 2\Phi'(|A|^2)\bigl(|B_j|^2-\Re(\overline A C_j)\bigr)
 +4\Phi''(|A|^2)\bigl(\Im(\overline A B_j)\bigr)^2\right]\dd g.
\]

Now, let us subtract this from \eqref{eq:integral} to 
cancel the $|B_j|^2$ terms, and \eqref{eq:integral} becomes 
\begin{equation}\label{eq:substract}
  \sum_j\int_G \left[ 2\Phi'(|A|^2)\Re(\overline A C_j)
   +4\Phi''(|A|^2)\Re(\overline A^{\,2}B_j^2)\right]\dd g.
\end{equation}
where the second term uses the elementary identity
$(\Re z)^2-(\Im z)^2=\Re(z^2)$.

Recall that Lemma~\ref{lem:casimir} gives 
\[
 \sum_jB_j^2=a_\lambda A^2,\qquad
 \sum_jC_j=c_\lambda A.
\]
Plugging this into \eqref{eq:substract} completes the proof of \eqref{eq:generalvariation}.
\end{proof}

\begin{remark}\label{rem:representation}
The unitary second variation $e^{\ii t T_j}w$ is used only for cancellation in the proof above. A better way to see the cancellation is to 
use the perturbation $e^{tT_j}w$ instead of $w+tT_j w$.
\end{remark}

For normalization, put $r_j(t)=\norm{w+tT_jw}^2$.
\begin{lemma}\label{lem:normvariation}
Let $w$ be a unit vector and $\Phi$ be of class $C^2$ near $1$.
For each $j$,
\begin{equation}\label{eq:normindividual}
 \left.\frac{\mathrm d^2}{\mathrm dt^2}\right|_0\Phi(r_j(t))
 =2\ip w{T_j^2w}\Phi'(1)+4\ip w{T_jw}^2\Phi''(1).
\end{equation}
Consequently,
\begin{equation}\label{eq:normcomposition}
 \sum_j\left.\frac{\mathrm d^2}{\mathrm dt^2}\right|_0\Phi(r_j(t))
 =2c_\lambda\Phi'(1)+4\bigl(c_\lambda-\mathcal U(w)\bigr)\Phi''(1).
\end{equation}
\end{lemma}
\begin{proof}
Expanding $r_j(t)=1+2t\ip w{T_jw}+t^2\ip w{T_j^2w}$ gives
$r_j(0)=1$, $r_j'(0)=2\ip w{T_jw}$, and
$r_j''(0)=2\ip w{T_j^2w}$. The chain rule proves the first identity.
Summing and using
\[
 \sum_j\ip w{T_j^2w}=c_\lambda,\qquad
 \sum_j\ip w{T_jw}^2=c_\lambda-\mathcal U(w)
\]
proves the second, by Lemma~\ref{lem:casimir} and \eqref{eq:uncertainty}.
\end{proof}

\section{Proof of the main results}\label{sec:proof}

We prove Theorem~\ref{thm:main} for two classes of functions separately.

\subsection{The power function \texorpdfstring{$x^p$}{x\textasciicircum p}}\label{sec:powers}

Fix a real $p>0$ with $p\ne1$ and define
\begin{equation}\label{eq:IN}
 I_p(w)=\int_G |\ip{v_g}{w}|^{2p}\dd g.
\end{equation}
The function $I_p$ is homogeneous of degree $2p$,
$G$-invariant, and phase-invariant.

For fixed $p>0$ with $p\ne1$, a vector $w$, and an index $j$, put
$r_j(t)=\norm{w+tT_jw}^2$ and set
\begin{equation}\label{eq:scalarcurves}
 H_j(t)=I_p(w+tT_jw),\qquad n_j(t)=r_j(t)^p.
\end{equation}

In the following, we only prove the case $p>1$. The case $0<p<1$ is similar, with an additional issue that $z\mapsto |z|^{2p}$ is not $C^2$ at zero; this technical issue is addressed in Appendix~\ref{app:smallpowers}.

\begin{proposition}\label{prop:powervariation}
For fixed $p>1$ and $w$, the functions in \eqref{eq:scalarcurves} satisfy
\begin{equation}\label{eq:DI}
 \sum_j H_j''(0)
   =\bigl(2pc_\lambda+4p(p-1)a_\lambda\bigr)I_p(w).
\end{equation}
For unit $w$, we also have
\begin{equation}\label{eq:DN}
 \sum_j n_j''(0)=2pc_\lambda+4p(p-1)\bigl(c_\lambda-\mathcal U(w)\bigr).
\end{equation}
\end{proposition}

\begin{proof}
We apply lemma~\ref{lem:scalarvariation} for $\Phi(x)=x^p$.
Substituting $\Phi'(x)=px^{p-1}$ and
$\Phi''(x)=p(p-1)x^{p-2}$ in \eqref{eq:generalvariation} gives
\[
 \sum_jH_j''(0)
 =\bigl(4p(p-1)a_\lambda+2pc_\lambda\bigr)I_p(w).
\]
For unit $w$, apply Lemma~\ref{lem:normvariation} to the same scalar
function. Since $\Phi'(1)=p$ and $\Phi''(1)=p(p-1)$,
\eqref{eq:normcomposition} gives \eqref{eq:DN} directly.
\end{proof}

Now we are ready to prove the main theorem for the power function.

\begin{proof}[Proof of Theorem~\ref{thm:main}: the power case]
By compactness, $I_p$ attains its maximum $M$ at a unit vector $\wopt$.
The value $M$ is positive because $\int_G Q_{\wopt}(g)\dd g=1/d$.
Use \eqref{eq:scalarcurves} with $w=\wopt$.  For each fixed index $j$, set
\[
 \phi_j(t)=H_j(t)-M n_j(t)
       =I_p(\wopt+tT_j\wopt)-M\norm{\wopt+tT_j\wopt}^{2p}.
\]
Homogeneity and maximality give $\phi_j(t)\le \phi_j(0)=0$ for every real $t$.  Hence $\phi_j''(0)\le0$.  Summing over $j$ and
using \eqref{eq:DI} and \eqref{eq:DN}, we obtain
\[
 0\ge\sum_j \phi_j''(0)
 =4p(p-1)M\bigl(\mathcal U(\wopt)-\mathcal U_{\min}\bigr)\ge0.
\]
Thus $\mathcal U(\wopt)=\mathcal U_{\min}$, and Lemma~\ref{lem:uncertainty} gives
$\wopt=e^{\ii\theta}v_h$ for some $\theta\in\R$ and $h\in G$.  Conversely, every coherent
vector gives equality.
\end{proof}

\subsection{The logarithmic function \texorpdfstring{$x\log x$}{x log x}}\label{sec:shannon}

For $w\ne0$, define
\begin{equation}\label{eq:logfunctional}
 F(w)=\int_G|\ip{v_g}{w}|^2\log|\ip{v_g}{w}|^2\dd g.
\end{equation}
The resolution of identity gives $\int_G|\ip{v_g}{w}|^2\dd g=\norm w^2/d$, hence
\begin{equation}\label{eq:normalizedlog}
 F\!\left(\frac{w}{\norm w}\right)
 =\frac{F(w)}{\norm w^2}-\frac1d\log\norm w^2.
\end{equation}
The function $x\log x$ is not twice differentiable at zero.
Lemma~\ref{lem:regularity} in Appendix~\ref{app:regularity} justifies
the differentiations used below.

\begin{proposition}\label{prop:variation}
For every unit vector $w$,
\begin{equation}\label{eq:logvariation}
 \sum_j\left.\frac{\mathrm d^2}{\mathrm dt^2}\right|_0F(w+tT_jw)
 =2c_\lambda F(w)+\frac{2c_\lambda+4a_\lambda}{d}.
\end{equation}
\end{proposition}
\begin{proof}
By the argument in the proof of Lemma~\ref{lem:regularity}, the differentiation steps in the proof of
Lemma~\ref{lem:scalarvariation} are valid for $\Phi(x)=x\log x$.
The same calculation gives
\begin{align*}
 \sum_j\left.\frac{\mathrm d^2}{\mathrm dt^2}\right|_0F(w+tT_jw)
 &=\int_G\left[4a_\lambda Q_w
       +2c_\lambda Q_w(1+\log Q_w)\right]\dd g\\
 &=2c_\lambda F(w)+\frac{2c_\lambda+4a_\lambda}{d},
\end{align*}
where products at $Q_w=0$ are interpreted by continuity and
$\int_GQ_w\dd g=1/d$.
\end{proof}

Now we are ready to prove the main theorem for the logarithmic function.
The inequality also follows directly from the power function result by a standard limit argument.

\begin{proof}[Proof of Theorem~\ref{thm:main}: the logarithmic case]
Continuity of $x\log x$ on $[0,1]$ makes $F$ continuous on the compact
unit sphere. Let $\wopt$ be a maximizer and write $M=F(\wopt)$.
For each fixed $j$, put
\[
 H_j(t)=F(\wopt+tT_j\wopt),\qquad r_j(t)=\norm{\wopt+tT_j\wopt}^2.
\]
By definition, 
\begin{equation}\label{eq:normalizedlogcurve}
F\!\left(\frac{\wopt+tT_j\wopt}{\norm{\wopt+tT_j\wopt}}\right)
       =\frac{H_j(t)}{r_j(t)}-\frac1d\log r_j(t)\le M. 
\end{equation}
Equivalently,  
\[
f_j(t):=H_j(t)-\frac{1}{d}r_j(t)\log r_j(t)- M r_j(t)\le 0.
\]
This gives $f_j''(0)\le 0$ for all $j$. By
Lemma~\ref{lem:normvariation} and Proposition~\ref{prop:variation}, 
\[
\sum_j f_j''(0)=2c_\lambda M+\frac{2c_\lambda+4a_\lambda}{d}
-2c_\lambda(\frac{1}{d}+M)-\frac{4}{d}\bigl(c_\lambda-\mathcal U(\wopt)\bigr)
=\frac4d\bigl(\mathcal U(\wopt)-\mathcal U_{\min}\bigr).
\]
All combined, we obtain 
\begin{equation}\label{eq:entropycontact}
 0\ge\sum_jf_j''(0)
   =\frac4d\bigl(\mathcal U(\wopt)-\mathcal U_{\min}\bigr)\ge0
\end{equation}
and thus $\mathcal U(\wopt)=\mathcal U_{\min}$.
Equality in the uncertainty relation Lemma~\ref{lem:uncertainty} implies that $\wopt$
must be coherent; conversely, every coherent vector gives equality by group invariance.
This proves \eqref{eq:functional} for $\Phi(x)=x\log x$.
\end{proof}

\appendix
\section{Power variations for \texorpdfstring{$0<p<1$}{0<p<1}}\label{app:smallpowers}

Here we provide missing details handling the technical issue for the case $0<p<1$. For a real function $f$
defined near zero, write
\[
 D_{\mathrm{sym}}^2f(0):=
 \lim_{t\to0}\frac{f(t)+f(-t)-2f(0)}{t^2}
\]
whenever this limit exists. 

\begin{lemma}\label{lem:smallpowervariation}
For $w\ne0$ and $H_j(t)=I_p(w+tT_jw)$, the symmetric second
variations exist and satisfy
\begin{equation}\label{eq:smallpowervariation}
 \sum_jD_{\mathrm{sym}}^2H_j(0)
 =\bigl(2pc_\lambda+4p(p-1)a_\lambda\bigr)I_p(w).
\end{equation}
\end{lemma}
\begin{proof}
Write $A=\ip{v_g}{w}$, $B_j=\ip{v_g}{T_jw}$,
and $S=\sum_j|B_j|^2$.
For a smooth real function $\Phi$, set
\[
 U_j(t)=\int_G\Phi\bigl(|\ip{v_g}{e^{\ii tT_j}w}|^2\bigr)\dd g.
\]
By Haar invariance, $U_j$ is constant. Put $C_j=\ip{v_g}{T_j^2w}$.
Summing the identities $U_j''(0)=0$ gives
\[
 0=\int_G\Big[2\Phi'(|A|^2)
       \bigl(S-\Re(\overline A\sum_jC_j)\bigr)
       +4\Phi''(|A|^2)\sum_j
                    \bigl(\Im(\overline AB_j)\bigr)^2\Big]\dd g.
\]
By Lemma~\ref{lem:casimir}, we have $\sum_jC_j=c_\lambda A$ and
\[
 2\sum_j\bigl(\Im(\overline AB_j)\bigr)^2
 =|A|^2S-\Re\bigl(\overline A^{\,2}\sum_jB_j^2\bigr)
 =|A|^2S-a_\lambda|A|^4.
\]
Substituting these identities, dividing by $2$, and rearranging yields
\[
 \int_G\bigl(\Phi'(|A|^2)+|A|^2\Phi''(|A|^2)\bigr)S\dd g
 =\int_G\bigl(c_\lambda|A|^2\Phi'(|A|^2)
       +a_\lambda|A|^4\Phi''(|A|^2)\bigr)\dd g.
\]
Apply this with $\Phi(x)=(x+\varepsilon)^p$,
where $\varepsilon>0$, and divide by $p$. Since
$\varepsilon+p|A|^2\ge p(|A|^2+\varepsilon)$ and $p-1<0$, we obtain
\begin{align*}
 p\int_G(|A|^2+\varepsilon)^{p-1}S\dd g
 &\le\int_G(|A|^2+\varepsilon)^{p-2}(\varepsilon+p|A|^2)S\dd g\\
 &=c_\lambda\int_G|A|^2(|A|^2+\varepsilon)^{p-1}\dd g\\
 &\quad +(p-1)a_\lambda\int_G|A|^4(|A|^2+\varepsilon)^{p-2}\dd g\\
 &\le c_\lambda I_p(w).
\end{align*}
Monotone convergence gives $S=0$ almost everywhere on $\{A=0\}$ and
\begin{equation}\label{eq:smallpowerbound}
 \int_G|A|^{2p-2}S\dd g
 \le\frac{c_\lambda}{p}I_p(w)<\infty.
\end{equation}
Here and below, singular integrands are defined to be zero on $\{A=0\}$.

For $f(z)=|z|^{2p}$, we claim the elementary estimate
\begin{equation}\label{eq:smallpowerdifference}
 |f(z+h)+f(z-h)-2f(z)|
 \le 2|h|^2|z|^{2p-2},\qquad z\ne0.
\end{equation}
Indeed, concavity and the parallelogram identity give
\[
 |z+h|^{2p}+|z-h|^{2p}
 \le2(|z|^2+|h|^2)^p
 \le2|z|^{2p}+2p|z|^{2p-2}|h|^2.
\]
The second inequality uses $(1+s)^p\le1+ps$ for $s\ge0$,
with $s=|h|^2/|z|^2$.
For the lower bound, write $x_+=\max\{x,0\}$. The arithmetic--geometric
mean inequality and $|z+h|\,|z-h|=|z^2-h^2|$ yield
\begin{align*}
 |z+h|^{2p}+|z-h|^{2p}
 &\ge2\bigl(|z+h|\,|z-h|\bigr)^p\\
 &\ge2(|z|^2-|h|^2)_+^p\\
 &\ge2|z|^{2p}-2|z|^{2p-2}|h|^2,
\end{align*}
where the last step uses $(1-s)_+^p\ge1-s$ for $s\ge0$.
Combining the two bounds proves \eqref{eq:smallpowerdifference}.
Taking $z=A$ and $h=tB_j$ in \eqref{eq:smallpowerdifference}
bounds the symmetric difference quotient by
$2|A|^{2p-2}|B_j|^2$, which is integrable by
\eqref{eq:smallpowerbound}. Dominated convergence therefore gives
\[
 D_{\mathrm{sym}}^2H_j(0)
 =\int_{\{A\ne0\}}\left[
 2p|A|^{2p-2}|B_j|^2
 +4p(p-1)|A|^{2p-4}\bigl(\Re(\overline AB_j)\bigr)^2\right]\dd g.
\]

For $\Phi_\varepsilon(x)=(x+\varepsilon)^p$, the corresponding
second-variation integrands are bounded in absolute value by
$C_p|A|^{2p-2}|B_j|^2$, uniformly in $\varepsilon>0$ on $\{A\ne0\}$,
and vanish almost everywhere on $\{A=0\}$.
Thus dominated convergence identifies their integrals in the limit
with $D_{\mathrm{sym}}^2H_j(0)$.
Apply Lemma~\ref{lem:scalarvariation} to $\Phi_\varepsilon$ and
let $\varepsilon\downarrow0$. Its right-hand side is dominated by
a constant times $|A|^{2p}$, so this proves \eqref{eq:smallpowervariation}.
\end{proof}

\section{The Logarithmic Variation}\label{app:regularity}

The main difficulty in the logarithmic case is justifying the second
differentiation under the integral; the following lemma addresses this issue.

\begin{lemma}\label{lem:log-regularity}\label{lem:regularity}
Fix $w\ne0$ and $j$, and put $A=\ip{v_g}w$, $B_j=\ip{v_g}{T_jw}$.
The function $H_j(t)=F(w+tT_jw)$ is twice differentiable at zero, with
\begin{equation}\label{eq:logaffineindividual}
 H_j''(0)=\int_G\left[2(1+\log|A|^2)|B_j|^2
       +4\frac{(\Re(\overline A B_j))^2}{|A|^2}\right]\dd g.
\end{equation}
\end{lemma}
\begin{proof}
Put $f(z)=|z|^2\log|z|^2$, with $f(0)=0$, and
$L(z)=z(1+\log|z|^2)$, with $L(0)=0$.
The function $f$ is $C^1$ on $\C$. Its pointwise parameter derivative is
\[
 \frac{\mathrm d}{\mathrm dt}f(A+tB_j)
 =2\Re\bigl(\overline{L(A+tB_j)}B_j\bigr).
\]
For $|t|\le\delta$, choose
$R\ge\max\{1,\norm w+\delta\norm{T_jw}\}$, so that
$|A|,|A+tB_j|\le R$. The derivative above is bounded uniformly in $t,g$ by
\[
 2\norm{T_jw}\sup_{0\le r\le R}r\,|1+\log r^2|<\infty,
\]
where the value at $r=0$ is interpreted by continuity.
Since Haar measure is finite, differentiation under the integral gives
\[
 H_j'(t)=2\Re\int_G\overline{L(A+tB_j)}B_j\dd g.
\]

We now prove the formula for $H_j''(0)$. To this end, we first show
that, for $|z|,|z+h|\le R$ and $z\ne0$,
\begin{equation}\label{eq:logincrement}
 |L(z+h)-L(z)|\le C_R|h|(1+|\log|z||).
\end{equation}
If $z+h=0$, then $h=-z$ and
\[
 |L(z+h)-L(z)|=|z|\,|1+\log|z|^2|
 \le 2|h|(1+|\log|z||).
\]
If $|z+h|\ge|z|>0$, the identity
\[
 L(z+h)-L(z)
 =h(1+\log|z+h|^2)+2z\log\frac{|z+h|}{|z|}
\]
and $\log x\le x-1$ give
\begin{align*}
 |L(z+h)-L(z)|
 &\le |h|\,|1+\log|z+h|^2|+2(|z+h|-|z|)\\
 &\le |h|\bigl(|1+\log|z+h|^2|+2\bigr).
\end{align*}
If $0<|z+h|<|z|$, interchange the two endpoints. Thus, for $z+h\ne0$,
\[
 |L(z+h)-L(z)|
 \le |h|\left(2+\left|1+\log\max\{|z|^2,|z+h|^2\}\right|\right).
\]
Since $|z|^2\le\max\{|z|^2,|z+h|^2\}\le R^2$, the last factor
is at most $3+2\log R+2|\log|z||$, proving \eqref{eq:logincrement}.

To obtain an integrable bound, put $S=\sum_j|\ip{v_g}{T_j w}|^2$.
Set $M=1+\norm w^2$. For $0<\varepsilon\le1$, the weighted identity
in the proof of Lemma~\ref{lem:smallpowervariation}, applied with
$\Phi_\varepsilon'(x)=\log(M/(x+\varepsilon))$, gives
\[
 \int_G\log\frac{M}{|A|^2+\varepsilon}S\dd g
 \le\int_G S\dd g
   +c_\lambda\int_G|A|^2\log\frac{M}{|A|^2}\dd g<\infty.
\]
The right-hand side is finite and independent of $\varepsilon$, since
$x\log(M/x)$ extends continuously at zero. Monotone convergence
as $\varepsilon\downarrow0$ gives $S=0$ almost everywhere on $\{A=0\}$
and shows that $(1+|\log|A|^2|)S$ is integrable on $\{A\ne0\}$.
By \eqref{eq:logincrement}, the difference quotient of the integrand
in $H_j'(t)$ at zero is bounded by
$C_R|B_j|^2(1+|\log|A||)$ on $\{A\ne0\}$, and vanishes almost
everywhere on $\{A=0\}$. Dominated convergence proves
\eqref{eq:logaffineindividual}, with its integrand set to zero on $\{A=0\}$.
Exactly the same argument applies to the affine curve
$\int_G f(A+\ii tB_j)\dd g$.

\end{proof}

\newcommand{\etalchar}[1]{$^{#1}$}

\end{document}